\documentclass[a4paper,12pt]{extarticle}

\usepackage{titlesec}

\titleformat{\paragraph}
{\normalfont\normalsize\bfseries}{\theparagraph}{1em}{}
\titlespacing*{\paragraph}
{0pt}{3.25ex plus 1ex minus .2ex}{1.5ex plus .2ex}

\usepackage{hyperref}
\usepackage{amsmath}

\usepackage{enumerate}

\usepackage{amsthm}
\usepackage{amsfonts}
\usepackage{subeqnarray}
\usepackage{mathrsfs}
\usepackage{latexsym}
\usepackage{dsfont}
\usepackage{mathtools} 
\usepackage{upgreek}
\usepackage{stmaryrd}

\usepackage{color}

\usepackage[margin=2cm]{geometry}
\allowdisplaybreaks[4]

\newcommand{\ud}{\mathrm{d}}
\newcommand{\beq}{\begin{equation}}
\newcommand{\eeq}{\end{equation}}
\newcommand{\ba}{\begin{eqnarray}}
\newcommand{\ea}{\end{eqnarray}}

\newcommand\Peq{P_{\text{eq}}}

\newcommand{\R}{\mathbb{R}}
\newcommand{\N}{\mathbb{N}}
\newcommand{\E}{\mathbb{E}}

\newcommand{\F}{\mathcal{F}}
\newcommand{\sL}{\mathrm{L}}
\newcommand{\sH}{\mathrm{H}}
\newcommand{\sU}{\mathrm{U}}
\newcommand{\sV}{\mathrm{V}}

\newcommand{\1}{\mathds{1}} \newcommand{\dd}{\mathrm{d}}
\newcommand{\dds}{\mathrm{d}s}
\newcommand{\ddr}{\mathrm{d}r}
\newcommand{\ddt}{\mathrm{d}t}

\newcommand{\ddx}{\mathrm{d}x}

\newcommand{\Proof}{\begin{proof}}
\newcommand{\EProof}{\end{proof}}

\everymath{\displaystyle}

\newtheorem{theorem}{Theorem}[section]

\newtheorem{proposition}{Proposition}[section]
\newtheorem{lemma}{Lemma}[section]

\newtheorem{definition}{Definition}[section]

\title{A New Non-Linear Density Fluctuations Stochastic Partial Differential Equation With a Singular Coefficient of Relevance to Polymer Dynamics and Rheology: Discussions, Proofs of Solution Existence, Uniqueness,  and a Conjecture }

\author{
Ludovic Gouden\`ege$^1$ and Liviu Iulian Palade$^2$ 
\footnote{E-mail: goudenege@math.cnrs.fr;   liviu-iulian.palade@insa-lyon.fr }
}

\numberwithin{equation}{section}

\begin{document}

\maketitle

\begin{flushleft}

$^1$ CNRS, F\'ed\'eration de Math\'ematiques de CentraleSup\'elec\\ B\^at. Bouygues, 9 rue Joliot Curie, 91190 Gif-sur-Yvette, France. 

$^2$ Universit\'e de Lyon, CNRS, Institut Camille Jordan UMR 5208, INSA-Lyon\\ P\^ole de Math\'ematiques,  B\^at. Leonard de Vinci No. 401, 21 avenue Jean Capelle, 69621 Villeurbanne, France.

\end{flushleft}

\begin{abstract}
In this paper we consider an entirely new  -  previously unstudied to the best of our knowledge - type of density fluctuations stochastic partial differential equation with a singular coefficient involving the inverse of a probability density.  The equation was recently introduced by Schieber \cite{jay3} while working on a new  polymer molecular dynamics approach that pertains to the generally called polymer reptation (aka tube) theory.  The corresponding probability density is the solution of an evolution equation (a stochastic transport) defined on a dynamical one-dimensional subspace.  A peculiarity of the here studied equation is its very singular pattern, even though it exhibits a well-posed structure.  As a first step towards furthering the understanding of this new class of equations we prove the existence and uniqueness of solutions under suitable conditions.  Moreover, capitalizing on the assumed gradient structure, we prove that the existing  solutions converge to an equilibrium density.
\end{abstract}

\begin{flushleft}

 


\end{flushleft}


\section{Introduction}\label{intro}

Basically there are two dominant, authoritative theories (together with their variants) in the realm of polymer melt molecular dynamics which are of paramount importance to the rheology of those non-Newtonian fluids.  Specifically, that of Bird, Curtiss, Armstrong and Hassager ~\cite{bird2} based on early ideas set-up by Kirkwood ~\cite{kj1}, and that of Doi and Edwards ~\cite{de1} which incorporates de Gennes' reptation concept ~\cite{dg}.  The former is known as  \textit{kinetical theory for polymer dynamics} while the later is commonly referred to as either the \textit{tube model}  or the \textit{reptation theory}.  In this paper we are concerned uniquely by the second approach, and for sake of clarity we only give a succinct overview of the key ideas and borrow some inspiration from a few of the acclaimed references like e.g. ~\cite{de1}, ~\cite{hui1}, ~\cite{lar}, ~\cite{mks}, ~\cite{faith},  ~\cite{yhl}, to which we equally refer the interested reader for detailed presentations. 

In a melt polymer material sample, the macromolecular chains are highly entangled with each other, hence the motion of a single chain is strongly influenced by the collective presence and dynamics of the surrounding neighbours.  Such an influence is geometrically idealized by considering an individual chain as being confined to a tube with deformable walls - hence the name of the theory.  Because the tube largely hinders the transverse motion, the large scale geometric rearrangement of the chain configuration takes place inside the tube walls via a snake-like reptational motion along its own length.  As the chain is geometrically approximated as a concatenation of deformable segments (one encounters a many body problem), in the de Gennes - Doi - Edwards theory a chain segment is subjected to a three-stage consecutive motion: a sudden deformation immediately followed by a quick process of retraction to its original length, process continued by a slower reptation motion of relaxation out of the original tube.  Actually the reptational diffusion also affects the shape of the both tube extremities in this way: the part of the tube which is emptied as the chain retracts will cease to exist (thus a constraint is released) while a new tube portion is formed at the other extremity and surrounds the diffusing chain.   Next, all these dominant interactions are accounted for in a so-called configurational  probability density diffusion equation - typically a parabolic partial differential equation (which may be deterministic or stochastic in nature).  Its solution constitutes the key ingredient for obtaining the stress tensor which ultimately enters the momentum balance system of equations describing the macroscopic fluid motion.  For some recent modeling results of the tube model in applied rheology see ~\cite{chp1}, ~\cite{chp2}, ~\cite{lip1} and references cited therein. 

The reptation theory (together with the kinetical approach) is popular in the area of polymer rheology.  Since it was published in its original form in 1978-79 (see ~\cite{de1}) it has been subjected  to a multitude of modifications, refinements, improvements to make it meet the needs of better predicting an ever wider variety of rheological experimental data; in turn this however made modeling necessarily more complex.  Such undertakings include - but not restricted to - tube length diameter fluctuations, constrained release, fluctuating slip links and slip springs, inhomogeneities associated with density fluctuations.  Many of these theoretical improvements were tested via advanced Monte Carlo simulations performed on stochastic equations. 

For instance, while focusing on the density fluctuations and their relation to quantities like slip spring distribution (see e.g. ~\cite{jay1}, ~\cite{jay2}), Schieber ~\cite{jay3} has recently authored a density fluctuations stochastic partial  differential equation, \textit{a simplified version of which} is the focus of this work: see the below given problem \eqref{j1}-\eqref{j5}.  All quantities are dimensionless.  

Let $\rho(t,s)$ be a stochastic unknown process, with $s\in \left[L_-,L_+ \right]$ being the space variable (related to a tube length being equal to $2L$) and $t\in[0,T]$ being the time variable.
Let $\mu(t,s)$ be the chemical potential, $W(t,s)$ a Wiener process (possibly a space-time white noise, a time white noise or a finite-dimensional brownian motion).

We are interested in the following formal stochastic partial differential equation
\beq\label{j1}
\partial_{t} \rho = \dfrac{1}{2}\dfrac{\partial}{\partial s}\left(\dfrac{1}{\rho} \dfrac{\partial \mu}{\partial s}\right) \ud t + \dfrac{\partial}{\partial s}\left(\dfrac{\ud W}{\sqrt{\rho}}  \right),\quad s\in \left[L_-,L_+ \right]  
\eeq
where the boundary conditions are stochastic in nature:

\beq\label{j2}
\ud L_{\pm}=-\dfrac{1}{\rho}\left[\dfrac{1}{2\rho}\left(\dfrac{\partial \mu}{\partial s} \right)_{\mathcal{T}} \ud t+\dfrac{\ud W}{\sqrt{\rho}} \right] \Bigg|_{s=L_{\pm}}.
\eeq
As is common in thermodynamics, the $(\cdots)|_{\mathcal{T}}$ notation indicates parameter (temperature) $\mathcal{T}$ is kept constant, we will skip this notation later.

The free energy (Hamiltonian) $\mathcal{E}$ of a generic unknown $\rho$ is given by 
\beq\label{j3}
\mathcal{E}(\rho)
=\int_{[L_{-},L_{+}]}\left[\rho(\cdot, s)+\dfrac{1}{\rho(\cdot, s)} \right] \ud s
=\int_{[L_{-},L_{+}]} V(\rho(\cdot, s)) \ud s,
\eeq
with $V$ the obvious potential. 

The chemical potential $\mu$ is usually given by the functional derivative of the free energy
\beq\label{j4}
\mu=\frac{\delta \mathcal{E}}{\delta\rho}=1-\dfrac{1}{\rho^2} \quad(+C)
\eeq
but the above chemical potential may be defined up to an additive constant $C$ which is used as a Lagrangian multiplier in various models.
However its presence here is impactless for only the gradient of the chemical potential will be used.

Equation \eqref{j1} is compatible with total mass conservation law and is also expected to exhibit an equilibrium probability density (invariant measure) $P_{eq}$ given by 

\beq\label{j5}
\ud \Peq(\rho)=\frac{1}{Z}\exp \left[-\displaystyle \mathcal{E}(\rho) \right]\mathds{1}_{\rho\geq 0} \ud\gamma(\rho)
\eeq
where $\gamma$ is a reference measure in some space, and $Z$ is a normalization constant.

The originally proposed free energy \cite{jay3} contained the extra $ \alpha \left(\dfrac{\partial \rho(s)}{\partial s} \right)^2 $ term in the integrand of equation \eqref{j3} and the extra $-2\alpha \dfrac{\partial^2 \rho}{\partial s^2}$ term in equation \eqref{j4}, with $\alpha\geq 0 $ being a model parameter. 
In the full model, when a gradient term is added to the free energy, the above introduced singular problem pertains to the  Cahn-Hilliard-Cook family of equations (with a singular mobility, singular potential and singular multiplicative noise). 
Specifically:

\beq\label{j3bis}
\mathcal{E}_\alpha(\rho)=\int_{[L_{-},L_{+}]}\left[\rho(\cdot, s)+\dfrac{1}{\rho(\cdot, s)}  + \alpha \left(\dfrac{\partial \rho(\cdot, s)}{\partial s} \right)^2 \right] \ud s,
\eeq

\beq\label{j4bis}
\mu=\frac{\delta \mathcal{E}_\alpha}{\delta\rho}=1-\dfrac{1}{\rho^2}-2\alpha \dfrac{\partial^2 \rho}{\partial s^2} \quad (+C).
\eeq

In this paper, taking \eqref{j1}-\eqref{j5} as starting point, we undertake to study the $n$-dimensional generalized problem by considering a regularized expression of the full problem, using an abstract formulation which encodes all the difficulties in the classical approach of SPDEs.
We detail the expected concept of solution that could remain unaltered in the singular case.
We also explain why a singular version of the 1-dimensional problem is well-posed with the extra-$\alpha$-term in \eqref{j3bis}-\eqref{j4bis} but under the assumption of additive noise.


Our main result is the following existence theorem (with the meaning of the solution concept being detailed later in the text):


\begin{theorem}[{\bf Main Result}]\label{th:main}
Let $\rho_{0}$, defined on a given space domain, be  a smooth enough initial probability density.
\begin{itemize}
\item Given a regular abstract formulation (possibly in the $n$-dimensional space), then, for all $T>0$, there exists a unique solution $(\rho_{t})_{t\in[0,T]}$ to the equation \eqref{j1}.
Moreover, under strict monotonicity, for $T\rightarrow +\infty$, the distribution of the solution converges to an equilibrium density (invariant measure).\label{th:main1}
\item Given a singular abstract formulation in the $1$-dimensional space but with additive noise and the extra-term with $\alpha >0$, then, for all $T>0$, there exists a unique solution $(\rho_{t})_{t\in[0,T]}$ to the equation \eqref{j1}.  Moreover, the distribution of the solution converges to the unique equilibrium density (invariant ergodic measure) given by equation \eqref{j5} with $\gamma$ a reference Gaussian measure.\label{th:main2}
\end{itemize}
\end{theorem}

This paper is organized as following: Section \ref{sec:ppmr} introduces the concept of solutions based on the general theory of SPDEs.
We prove that some a priori estimates are valid for the introduced solutions.
It also contains the problem under scrutiny together with the appropriate physical considerations. 
In Section \ref{sec:regular-abstract-formulation}, 
we introduce a regular abstract formulation to present the essential difficulties which prevent us of using classical techniques about the existence of solutions.
In Section \ref{sec:simplified-formulation}, where the main result Theorem \ref{th:existencesimplified} is proved, we introduce a simplified formulation by circumventing selected difficulties which were presented in the previous section.
However, we can keep a singular potential which constitutes the core of the physical model.  Next, the path-wise uniqueness is given in Section \eqref{sec:uniqueness} for all the previous here considered variants of the original model.   The existence and convergence of invariant measures are presented in Section \eqref{sec:invariant}.
We eventually draw conclusions for this paper model, introduce a conjecture about the existence and uniqueness of solutions for a general problem in Section \eqref{concl}.

\section{Problem Presentation and Main Results}\label{sec:ppmr}

At first sight, this equation could be seen as a Cahn-Hilliard-Cook stochastic partial differential equation, and in simple cases, the existence and uniqueness could be obtained by known results.
The Cahn-Hilliard-Cook equation is a model that describes phase separation in a binary alloy (see \cite{Cahn}, \cite{CH1} and \cite{CH2}) in the presence of thermal fluctuations (see \cite{COOK} and \cite{LANGER}). It generally reads:
\beq
\label{e0.1}
\left\{\begin{array}{ll}
\partial_{t} \rho = \frac{1}{2}\nabla \cdot \left(M(\rho) \nabla \mu\right)+\dot{\xi}&\text{ on } \mathcal{D}\subset \R^n,\\
\\
\mu = V'(\rho) -2\alpha\Delta \rho,&\text{ on } \mathcal{D}\subset \R^n,\\
\\
\nabla \rho \cdot \nu = 0 = \nabla \mu \cdot \nu, &\text{ on } \partial\mathcal{D},\\ 
\end{array}\right.
\eeq
where $t$ denotes the time variable, $\alpha$ is a small non-negative parameter, $\nu$ is the outward normal on the boundary $\partial\mathcal{D}$ of the space domain $\mathcal{D}$, $\nabla \cdot$ and $\nabla$ are respectively the divergence and gradient operators in the space variable $s\in \mathcal{D}$, and $\Delta$ is the Laplace operator.
The variable $\rho$ represents the ratio between the two species and the noise term $\dot{\xi}$ accounts for the thermal fluctuations.
For this kind of equation, a key point is that the average $\int_\mathcal{D} \rho\ \dds$ is a conserved quantity.

In the Cahn-Hilliard model, the nonlinear term $V'$ is based on a double-well shaped 
potential $V$ (hence non-convex), and the function $M$ represents the mobility (a diffusion velocity).
Usually the mobility is a constant or the function $M:\rho \mapsto (1-\rho^2)$.
In the deterministic case, the equation is obtained by taking the gradient of the free energy $\mathcal{E}$ in $\mathrm{H}^{-1}(\mathcal{D})$ but on a fixed domain $\mathcal{D}$.

Since the gradient operates in $\mathrm{H}^{-1}(\mathcal{D})$, a fourth order equation is obtained.
Thus, if $\alpha>0$, the presented model is of Cahn-Hilliard-Cook type, but with moving domain, singular mobility and multiplicative noise.
Indeed, in this new model, the mobility
\begin{equation}\label{eq.M}
M : \rho \mapsto 1/\rho,
\end{equation}
is singular leading also to a singular term in front of the noise. Indeed, by fluctuation-dissipation theorem, we expect the square root of the mobility in front of the noise (see Equation \eqref{j1}).

This new model will lead to multiple mathematical difficulties, in particular concerning the multiplicative noise, which could be a hard question.
All of these difficulties will not be solved in this paper, but we will give hints for possible direction of improvements. The aim of this article is to give the first mathematical structure as a starting point for future research in {\em tube model} and {\em reptation theory}.

\medskip

The deterministic Cahn-Hilliard equation where $V$ is a polynomial function has first been studied in \cite{CH1}, \cite{LANGER} and \cite{MR763473}.
The drawback of a polynomial nonlinearity is that the solution is not constrained to remain in the physically relevant interval. Singular nonlinear terms such as the logarithmic or negative power nonlinearity do remedy this problem. 
Such non smooth functions $V$ have also been considered (see \cite{MR1123143} and \cite{MR1327930}).
This stochastic Cahn-Hilliard-Cook equation has also been first studied in the case of a polynomial nonlinearity: see \cite{BLMAWA1}, ~\cite{BLMAWA2},  ~\cite{MR1867082}, ~\cite{MR1897915}, ~\cite{MR1359472}, ~\cite{MR1111627}.
Again, with a polynomial nonlinear term the solutions do not in general remain in the physical space of interest.
It actually gets even worse for in the presence of noise any solution leaves this interval immediately with positive probability.
Of notice, however, in the work of Debussche and Gouden\`ege ~\cite{debgou, goudenege} and that of Debussche and Zambotti ~\cite{MR2349572}, the solution is forced to stay in the physical interval thanks to a singular potential.\medskip

In the case considered here, the parameter $\alpha$ may vanish, the mobility is highly singular, the stochastic term is a singular multiplicative noise, and the potential $V$ is strictly convex, thus the equation is not of Cahn-Hilliard-Cook type, and we will need a new abstract formulation to treat this new non-linear density fluctuations equation.
On the contrary, in the case of additive noise assumption, with the extra-$\alpha$-term and regular mobility, we can consider a potential $V(\rho)=\rho+1/\rho$ in the energy $\mathcal{E}_\alpha$ since it is related to the model described in \cite{goudenege}.

It would be preposterous to assume studying such a model is a classical undertaking.
Henceforth conspicuously obtaining a preliminary result with an abstract formulation or simplified term potential is all the more necessary and useful. 
We now refer the reader to the classical work of Da Prato and Zabczyk \cite{MR1207136} for notations, basic notions and concepts.

We describe the concept of solution for a regular potential $V$ (at least of class $\mathcal{C}^2$ and convex) and a regular mobility $M$ (at least of class $\mathcal{C}^1$ and non-negative).  Also observe the concept of solution (valid for smooth potential) may remain unaltered in the singular case, but integrability conditions should be changed to give a clear meaning to all involved terms.
In Section \ref{sec:regular-abstract-formulation}, we will simplify some assumptions and keep others unchanged.
In Section \ref{sec:simplified-formulation}, we will also simplify some aspects and keep the others unchanged, and the integrability conditions will be detailed.

\begin{definition}\label{def:sol}
Let $\ell_{-}$ and $\ell_{+}$ with $\ell_{-}< \ell_{+}$ be fixed. Let $\rho_{0} \in \mathcal{C}([\ell_{-},\ell_{+}];\R_{+}^{*})$ be the initial density in the space of probability measures $\mathcal{P}$.
We say that $\left(L_{-}(t), L_{+}(t), \rho(t), W_t\right)_{t\in[0,T]}$, defined on a filtered complete probability space $\left(\Omega, \mathbb{P}, \F, (\F_t)_{t\in [0,T]}\right)$, is a weak solution to \eqref{j1} on $[0,T]$ for the initial condition $\rho_{0}$ if  we have
\begin{enumerate}
\item[(a)] a.s. $L_{\pm} \in \mathcal{C}\left([0,T];\R\right)$ with $L_{-}(t)<L_{+}(t)$ for all $t\in[0,T]$ and $L_{\pm}(0) = \ell_{\pm}$,
\item[(b)] $\rho \in \sL^2\left(\Omega; \mathcal{C}\left([0,T];\mathcal{C}\left([L_{-},L_{+}]; \R_{+}^{*}\right)\right)\cap  \sL^2\left([0,T];\mathrm{H}^1\left([L_{-},L_{+}]; \R_{+}^{*}\right)\right)\right)$ and $\rho(0) = \rho_{0}$,
\item[(c)] $\mu = V'(\rho) \in \sL^2\left(\Omega; \mathcal{C}\left([0,T];\mathcal{C}^1\left([L_{-},L_{+}]; \R\right)\right)\right)$ and $\mu(0) = V' (\rho_0)$,
\item[(d)] $W$ is a continuous version of Brownian motion,
\item[(e)] the process $\left(\rho(t,\cdot),W_t\right)$ is $(\F_t)$-adapted,
\item[(f)] the mobility $M$ is such that $M(\rho) \in L^2\left(\Omega; \sL^1\left([0,T];\sL^1((L_{-},L_{+}); \R_+)\right)\right)$,
\item[(g)] the operator $Q$ is such that $Q(\rho) \in \sL^2\left(\Omega;\sL^2\left([0,T];\sL^2((L_{-},L_{+}); \R_+)\right)\right)$, also we assume that $\mathrm{Tr}(Q^*\Delta Q)$ is controlled by $C_{Q,1}+C_{Q,2}\|\rho\|^2$.
\item[(h)] for all $h \in 
\mathcal{C}^4(\R;\R)
$ and for all $0 \leq  t \leq T$ :
\begin{eqnarray*}
\langle \rho(t,\cdot),h\rangle &=& \langle \rho(0,\cdot),h\rangle - \frac{1}{2}\int_{0}^t \left\langle M(\rho(r,\cdot))\nabla \mu, \nabla h\right\rangle \ddr -  \int_{0}^t\left\langle \nabla h,Q(\rho(r,\cdot))\dd W_r\right\rangle\\
&& - \alpha \int_{0}^t \left\langle \rho(r,\cdot), \Delta^2 h\right\rangle \ddr.
\end{eqnarray*}
\end{enumerate}
\end{definition}\par\medskip

We point out that the continuity of $L_{\pm}$ allows to consider functions on compact sets which simplifies many technical points, and the scalar product are defined on ${[L_{-},L_{+}]}$ assuming the quantity vanishes outside this domain.
Since $\rho>0$, then its infimum is strictly positive, the chemical potential $\mu$ is well defined even if $V$ has singularities at point $0$, and $M(\rho)$ is $\sL^{1}$ for the same reason.
In the case of a Wiener process in abstract space $U$ at point $(d)$, the condition $(g)$ could be set as: $Q(\rho)$ is Hilbert-Schmidt from $U$ to the Hilbert space $\sH=\sL^2(\mathcal{D})$ with maximal domain $\mathcal{D}$ and inner product $\langle\cdot ,\cdot \rangle$.
Actually in Definition \eqref{def:sol} we provided details on all the needed properties to make sure all terms in the last point (h) clearly make sense.

We shall now, step by step, proceed to proving the afore announced main result.
Before engaging in this undertaking, we pause for several short clarifying comments.
On one hand, it is worthy stating the proof may not be particularly challenging if one considers the equation as being a system gradient in infinite dimension, by using the Da Prato and Zabczyk~\cite{MR1207136} framework and by additionally assuming the noise to be an additive space-time Wiener process.
In this case there will be a unique invariant measure, given by a Gibbs-like formula, which is compatible with the expected equilibrium probability density in equation \eqref{j5}.
On the other hand now, there are multiple obstacles to be dealt with as one proceeds along those lines. For instance, one is the gradient structure needing more clarifications.
Another difficulty concerns the moving boundary, since the space domain is not fixed.
This later one seems however not of a hindering nature as the imposed boundary condition has been constructed to cancel supplementary terms coming from integration by parts formula.

In order to explore the concept of solutions, we will detail some a priori estimates and facts based on the existence of smooth solutions.
It gives strong hints that the equation is well-posed and satisfied expected physical properties.

\begin{proposition}\label{th:Proposition}
Consider a solution in the sense of Definition \ref{def:sol} up to time $T>0$,
then the solution lives in $\sL^2(\Omega;\sL^\infty([0,T]; \sL^2([L_-,L_+])))$.
Moreover, with $\alpha=0$, for chemical potential $\mu$ given by \eqref{j4} and mobility $M$ given by equation \eqref{eq.M}, then the energy $\mathcal{E}$ of the solution is well defined for all $t\in [0,T]$.
\end{proposition}

\begin{proof}
Let $\alpha=0$, by It\^o's formula, we have
\begin{eqnarray*}
\frac{1}{2} \mathrm{d} \|\rho\|^2 &=& \langle \dd\rho, \rho \rangle\\
&=& -\frac{1}{2}\left\langle M(\rho)\nabla \mu, \nabla \rho \right\rangle\ddt - \left\langle \nabla \rho,Q(\rho)\dd W_t\right\rangle+\frac{1}{2}\mathrm{Tr}( Q^* (-\Delta) Q)\ddt\\
&=& -\frac{1}{2}\left\langle M(\rho)\nabla \rho V''(\rho), \nabla \rho \right\rangle\ddt - \left\langle \nabla \rho,Q(\rho)\dd W_t\right\rangle+\frac{1}{2}\mathrm{Tr}( Q^* (-\Delta) Q)\ddt\\
&\leq& - \left\langle \nabla \rho(t,\cdot),Q(\rho(t,\cdot))\dd W_t\right\rangle+(C_{Q,1}+C_{Q,2}\|\rho\|^2)\ddt
\end{eqnarray*}

using that $V$ is a convex potential and $M$ is a non-negative mobility, and 
 $C_{Q,1}$ and $C_{Q,2}$ are finite constants invoking that the operator $Q$ has enough boundedness properties (almost linear growing, Lipschitz assumption, etc.). Here, since the noise is a one-dimensional Brownian motion, the assumption in point (g) could be simplified.
But remark that the classical assumption about finite trace may be used here to generalize to Hilbert-Schmidt type operators and space-time Wiener processes.
Using that the stochastic term is a martingale with null expectation, by Gronwall's Lemma, we obtain a bound for all $t\in [0,T]$, which leads to a bound for the $\sL^\infty$ norm such that
\[
\mathbb{E}\sup_{t\in[0,T]}\|\rho(t,\cdot)\|^2
\leq 
\mathbb{E}\|\rho(0,\cdot)\|^2
+C_T = \|\rho_0\|^2 + C_T < +\infty
\]
for some positive constant $C_T$ possibly depending on $\rho_0, T, C_{Q,1}$ and $C_{Q,2}$.
At this point it is clear that the initial data $\rho_0$ can be randomized (with bounded second moment) to obtain exactly the same result.
If we have $\alpha>0$, the same computation leads to 
\[
\mathbb{E}\left[\sup_{t\in[0,T]}\|\rho(t,\cdot)\|^2 + 2\alpha \int_0^T \|\Delta \rho(t,\cdot)\|^2 \ddt\right]
\leq 
 \|\rho_0\|^2 + C_T < +\infty
\]
which is eventually a stronger result.

Concerning the energy, with $\alpha=0$, using that almost surely $\rho$ is a continuous non-negative function on compacts spaces, hence bounded for below by a non-negative constant $\varepsilon>0$, the It\^o's terms may be bounded, such that all the terms are well-defined using explicit truncation, and we obtain by It\^o formula
\begin{eqnarray*}
 \mathrm{d} \left(\rho + \frac{1}{\rho}\right) &=& \left\langle \dd\rho, 1-\frac{1}{\rho^2} \right\rangle = \langle \dd\rho, \mu\rangle\\
&=& -\frac{1}{2}\left\langle M(\rho)\nabla \mu, \nabla \mu \right\rangle\ddt - \left\langle \nabla \mu,Q(\rho)\dd W_t\right\rangle+\frac{1}{2}Tr( Q^*\nabla^* H(\rho) \nabla Q)\ddt
\end{eqnarray*}
where $H$ is the Hessian operator of the function $x\mapsto x+1/x$.
Using again that the mobility is non-negative, we obtain the following inequality
\begin{equation}\label{estim}
 \mathrm{d} \left(\rho + \frac{1}{\rho}\right) \leq - \left\langle \nabla \mu,Q(\rho)\dd W_t\right\rangle+C(Q,\rho,\varepsilon) \ddt
\end{equation}
with a random function $C(Q,\rho,\varepsilon)$ which depends quadratically on $Q$ and $\rho$, but certainly on the truncated parameters $\varepsilon$ in unbounded manner. Using Doob's inequality for the martingale term, the assumptions about integrability in Definition \ref{def:sol} will be enough to obtain
\[
\sup_{t\in[0,T]}\mathcal{E}(\rho(t,\cdot)) \leq \mathcal{E}(\rho(0,\cdot)) + \sup_{t\in[0,T]} \int_0^t\left|\left\langle \nabla \mu,Q(\rho)\dd W_t\right\rangle\right|+ \int_0^t \left|C(Q,\varepsilon,\rho)\right| \ddt < +\infty.
\]
\end{proof}

Assuming that we can take the expectation in Equation \eqref{estim} to suppress the martingale term, we obtain a bound in expectation such that
\[
\mathbb{E}\left[\sup_{t\in[0,T]}\mathcal{E}(\rho(t,\cdot))
\right]
\leq 
\mathbb{E}\left[\mathcal{E}(\rho_0)\right] + \mathbb{E}\left[\int_0^T C(Q,\varepsilon,\rho)\ddt\right] 
\]
which is not independent of the truncated parameter $\varepsilon$, thus potentially infinite when passing to the limit or taking expectation.
Remark again that the initial data can be randomized (with bounded inverse first moment, i.e. bounded energy) to obtain the same result.

Finally, in the case $\alpha>0$, the same computation leads to
\begin{eqnarray*}
&\mathbb{E}\left[\sup_{t\in[0,T]}\mathcal{E}(\rho(t,\cdot))
+ \frac12\inf_{t\in[0,T]} |M(\rho(t,\cdot))| \int_0^T \|\nabla \mu\|^2 \ddt
\right]\\
&\leq 
\mathcal{E}(\rho_0) + \mathbb{E}\left[\int_0^T C(Q,\varepsilon,\rho)\ddt\right]  + \alpha \left|\int_0^T \langle \Delta \rho, \Delta \mu\rangle \ddt\right|
\end{eqnarray*}
but the last term is not controlled by integrability condition given in Definition \ref{def:sol}, and it seems unreasonable to expect such controlability.

\newpage

\section{Regular abstract formulation}\label{sec:regular-abstract-formulation}

In this section, we will assume that the equation can be written in a classical abstract form
\[
\dd X_t + A(X) \dd t = B(X) \dd W_t
\]
with operators satisfying assumptions that will be detailed below.  Considering this form, the proof of existence and uniqueness of solution is based on classical theorems the detailed presentations of which are given in \cite{MR1207136}, \cite{GyongyKrylov} or \cite{PR}.
Precisely, we assume to have a Gelfand triple $(\sV,\sH,\sV^*)$, where $\sV$ is a reflexive Banach space and $\sH$ is an Hilbert space, and a space of Hilbert-Schmidt operators $\sL_{HS}(\sU,\sH)$ on some abstract space $\sU$ to define the Wiener process.
We write the operators on these spaces $A:\sV\rightarrow \sV^*$ and $B:\sV\rightarrow \sL_{HS}(\sU,\sH)$. The notation $\,_{\sV^*}\langle \cdot, \cdot \rangle_\sV$ is the canonical notation for the duality bracket between the spaces $\sV^*$ and $\sV$, compatible with the inner product $\langle \cdot, \cdot \rangle$ on the space $\sH$.
We assume that they satisfy the following assumptions:





\begin{itemize}
\item (Hemicontinuity) For all $u,v,w \in \sV$
\[
\lambda \mapsto 2\,_{\sV^*}\langle A(u+\lambda v), w\rangle_{\sV} \text{ is continuous}.
\]
\item (Weak Monotonicity) There exists $c\in\mathbb{R}$ such that for all $u,v\in \sV$
\[
2\,_{\sV^*}\langle A(u)-A(v), v-u\rangle_{\sV} + \|B(u) - B(v)\|_{HS}^2 \leq c \|u-v\|_\sH^2.
\]
\item (Coercivity) There exist $q \in\ ]1,\infty[$, $c_1\in\mathbb{R}$, $c_2\in ]0,\infty[$ and an $\mathcal{F}_t$-adapted process $f\in \sL^1(0,T)$ such that for all $u\in \sV$
\[
2\,_{\sV^*}\langle A(u), u\rangle_{\sV} + \|B(u)\|_{HS}^2 \leq c_1 \|u\|_\sH^2-c_2\|u\|^q_\sV +f(t).
\]
\item (Boundedness)
 There exist $c_3\in [0,\infty[$
 and an $\mathcal{F}_t$-adapted process $g\in \sL^{\frac{q}{q-1}}(0,T)$ such that for all $u\in \sV$
\[
\| A(u)\|_{\sV^*} \leq c_3 \|u\|^{q-1}_\sV +g(t).
\]
\end{itemize}

It is important to notice that many problems have been solved assuming that an abstract formulation actually exists.
Indeed, the moving domain $[L_-,L_+]$ will be certainly encoded in the concept of solution as $X_t$ being $(L_-, L_+, \rho_t)$.
Or, under reformulation, we could assume that the domain is fixed
and given by the initial domain $[\ell_-,\ell_+]$, which solves the point $(a)$.
Moreover, the fiber bundles in points $(b)$ and $(c)$ of Definition \ref{def:sol} are now well-defined as classical spaces.
Finally, the mobility $M$ and the operator $Q$ are included in the operator $A$ and $B$, which will ensure the points $(f)$ and $(g)$. 
Even the concept of chemical potential $\mu$ has been included in the operator $A$.
Within this framework, the concept of solution introduced in \ref{def:sol} will be simplified since we can define a new one. 

One has:

\begin{definition}\label{def:solabstract}
We say that $\left(X(t), W(t)\right)_{t\in[0,T]}$, defined on a filtered complete probability space $\left(\Omega, \mathbb{P}, \F, (\F_t)_{t\in [0,T]}\right)$, is a weak solution to \eqref{j1} on $[0,T]$ for the initial condition $X_{0}$ if  we have
\begin{enumerate}
\item[(b')] $X \in \sL^2\left(\Omega; \mathcal{C}\left([0,T];\sH\right)\cap  \sL^2\left([0,T];\sV\right)\right)$ and $X(0) = X_{0}$,
\item[(d')] $W$ is a continuous version of Brownian motion (or Wiener process),
\item[(e')] the process $\left(X(t),W(t)\right)$ is $(\F_t)$-adapted,
\item[(h')] for all $h \in \sV^*$ and for all $0 \leq  t \leq T$ :
\begin{eqnarray*}
\langle X(t),h\rangle &=& \langle X(0),h\rangle - \int_{0}^t \left\langle A(X(r)), h\right\rangle \ddr + \int_{0}^t\left\langle B(X(r)),h \dd W_r\right\rangle.
\end{eqnarray*}
\end{enumerate}
\end{definition}\par\medskip

\subsection{Existence of solution}

One has the following existence Theorem, standing actually for the first part of Theorem \ref{th:main1}.

\begin{theorem}\label{th:existenceabstract}
Suppose that there exists a regular abstract formulation satisfying classical assumptions. 
Given an initial data in the Banach space $\sV$, then, for all $T>0$, there exists a unique solution $(X_t)_{t\in[0,T]}$ in the sense of Definition \ref{def:solabstract} up to time $T$.
\end{theorem}


\begin{proof}
We have supposed that the classical assumptions are satisfied, thus we do not have to verify them. However, we will show how they may be obtained with assumptions on mobility $M$, potential $V$ and operator $Q$, or on regularized versions of these objects. 

Usually, to prove the existence of a solution of a singular equation, we need to prove that a regularized version of the equation is well-posed with estimates being independent of the regularization procedure.
The possible regularized version may be some finite dimensional approximation (like Galerkin), a truncation of the non-linearity, a Yosida regularization, or a numerical approximation.
Thus, all future approaches will have to use the following calculations and to prove that the estimates are independent of the regularization procedure.

In Theorem \ref{th:existenceabstract}, we proved that some solution could exist in a fixed domain and that the It\^o's formula (being used for a priori estimates) suggests that the operator $A$ has the form $A : X\mapsto \frac12\langle M(X)V''(X)\nabla X, \cdot \rangle$ plus extra additive terms involving Laplace operator when $\alpha\neq 0$. It also suggests that the operator $B$ has the form $B:X\mapsto (w \mapsto \langle \cdot ,\nabla^* Q(X) w \rangle )$.  Denote $a:=x\mapsto M(x)V''(x)$ the non-linearity (or, more precisely, a regularized version of it) to simplify the notations in the following computations. It is assumed that it is non-negative, as suggested by the expected physical properties of the convex potential $V$ and the non-negative mobility $M$.
We now illustrate the difficulties in obtaining the expected bounds.

\textbf{Hemicontinuity.}

Suppose that the function $a$ is continuous and bounded from below by non-negative constant (locally on balls of $\sV$).
Let $u,v,w \in \sV$. Let $(\lambda_k)_{k\in\mathbb{N}}$ converging to $\lambda$ as $k$ grows to infinity, then
\begin{align*}
&\left|2\,_{\sV^*}\langle A(u+\lambda_k v)- A(u+\lambda v), w\rangle_\sV\right|\\
=&\left|\langle a(u+\lambda_k v)\nabla (u+\lambda_k v)-a(u+\lambda v)\nabla (u+\lambda v) , \nabla w \rangle\right|\\
\leq&\left|\langle \left(a(u+\lambda_k v)-a(u+\lambda v)\right)\nabla (u+\lambda_k v) , \nabla w \rangle\right|
+\left|\langle a(u+\lambda v)(\lambda_k-\lambda)\nabla v , \nabla w \rangle\right|.
\end{align*}
Since the sequence $\lambda_k$ is bounded, thus the sequence $u+\lambda_k v$ lives in a ball $\mathcal{B}$ in $\sV$. On this ball, the function $a$ is continuous and bounded. Let $\epsilon>0$ be fixed, there exists $\delta_\epsilon>0$ small enough such that $|\lambda_k-\lambda|< \delta_\epsilon$ implies $|a(u+\lambda_k v)-a(u+\lambda v)|< \epsilon$
and thus for large index $k\in\N$ we have that
the first term is bounded by 
$\epsilon \sup_{x\in\mathcal{B}}\|x\|_\sV \| w\|_\sV$ and the second term is bounded by $\delta_\epsilon \ |a(u+\lambda v)|\ \|v\|_\sV \| w\|_\sV$, which ensures convergence to $0$ and continuity.\medskip

\textbf{Weak Monotonicity.} 

Let $u,v,w \in \sV$, then
\begin{align*}
&2\,_{\sV^*}\langle A(u)-A(v), v-u\rangle_V + \|B(u) - B(v)\|_{HS}^2\\
=&\langle a(u)\nabla u -a(v)\nabla v, \nabla (v-u) \rangle + \|\nabla^*Q(u) -\nabla^*Q(v)\|^2_\sH\\
=&-\int_{\mathcal{D}} a(u)|\nabla (u - v)|^2 + \int_{\mathcal{D}}  (a(u)-a(v))\nabla v \cdot \nabla (v-u) + \|\nabla u\  Q'(u) -\nabla v\  Q'(v)\|^2_\sH.
\end{align*}

At this point, in order to obtain the bound $c\|u-v\|^2_\sH$, we see that the first term is not a problem since $a$ is assumed to be non-negative. It can help to control other problematic terms.
For the second term, we certainly need some globally Lipschitz regularity for the function $a$, but which is not precisely clear at this point since mobility and potential are singulars.
For the third term, we certainly need the regularity and boundedness of the derivative of $Q$. Remark that in the additive case ($B$ is a constant), this term does not appear. 
Remark also that if it exists an extra term with $\alpha>0$ involving Laplace operator, then we have an other term with the good sign which can help to control the problematic terms.

\textbf{Coercivity.} 

The coercivity corresponds to a control at infinity of the equation.
The leading term is $c_1 \|u\|^2_H$, which should absorb the operator $\,_{\sV^*}\langle A(u),u\rangle_\sV $, $B$ or a supplementary  term $c_2\|u\|^q_\sV$. 

Consider the equation

\begin{align*}
2\,_{\sV^*}\langle A(u), u\rangle_\sV + \|B(u)\|_2^2 + c_2 \|u\|^q_\sV
=&-\langle a(u) \nabla u, \nabla u\rangle + \|\nabla u\ Q'(u)\|^2_\sH \\
&+ c_2 \langle \nabla u, \nabla u\rangle^q+ c_2 \langle u, u\rangle^q.
\end{align*}

We have to bound these terms by $c_1\|u\|^2_\sH$ with $c_1\in \R$ but remember that $c_2>0$.
It is clear that the term $c_1\|u\|^2_\sH$ will not control the terms involving $\|u\|_\sV$ since it would imply a Poincaré inequality in reversed order, thus the only term with the good sign is the first one. 
This term should control the term with the constant $c_2$ which forces the constant $c_2$ not to be large, and $a$ should be lower bounded by constant $c_2>0$. The second term involves also the gradient of $u$, thus it is a supplementary term to control.
The last term involves only the norm of $u$ in the Hilbert space $\sH$ thus it is controlled by $c_2 (1+\|u\|^2)$ at least when $1<q\leq 2$.

\textbf{Boundedness.} 

Recall that the norm in $\sV^*$ is defined as the supremum of the quantity $\,_{\sV^*}\langle A(u),v\rangle_\sV$ when $\|v\|_\sV\leq 1$. Let $v\in \sV$, we are looking for an estimates
\begin{align*}
\left|\,_{\sV^*}\langle A(u),v\rangle_\sV \right|
:=\frac12\left|\langle a(u) \nabla u, \nabla v\rangle\right|
\leq c_3 \|u\|_\sV^{q-1} + g.
\end{align*}
It is now clear that the constant $c_3$ should control the function $a$ which is expected to be bounded from above.
\end{proof}

\section{Simplified formulation but with singular potential}
\label{sec:simplified-formulation}

In this section, we can study other formulations which are related to the problem introduced earlier. 
In particular, if we assume to work in dimension 1 in space, the noise is additive and that we have the extra-term containing $\alpha>0$, we can prove the existence of unique solution even with a singular potential $V$ of the form $\rho + 1/\rho$.
Furthermore, using formal computations, we can derive an equivalent problem with a modified chemical potential of the form $-1/\rho^3$. This equation could also be solved in the same previous simplified framework.

Moreover, in both cases, the solution will converge to the unique equilibrium density (invariant measure) given by the formula \eqref{j5} with $\gamma$ a reference Gaussian measure.

First, assuming that all the following computations are legitimate, we can derive an equivalent formulation of the problem to obtain a similar form.

One has:

\begin{proposition}\label{prop:equiv}
Assume the solution is non-negative and continuously differentiable.  Then  the statement (h) in \ref{def:sol} is equivalent,
for all $h \in 
\mathcal{C}^4(\R;\R)
$ and for all $0 \leq  t \leq T$, to
\begin{eqnarray*}
\langle \rho(t,\cdot),h\rangle &=& \langle \rho(0,\cdot),h\rangle - \frac{1}{3}\int_{0}^t \left\langle \frac{1}{\rho(r,\cdot)^3}, \Delta h\right\rangle \ddr -  \int_{0}^t\left\langle \nabla h,Q(\rho(r,\cdot))\dd W_r\right\rangle\\
&& - \alpha \int_{0}^t \left\langle \rho(r,\cdot), \Delta^2 h\right\rangle \ddr.
\end{eqnarray*}
\end{proposition}

\begin{proof}
Assuming the solution is non-negative, the quantities $M(\rho) = \frac{1}{\rho}$ and $\mu = 1-\frac{1}{\rho^2}$ are well-defined. It will justify the study of the chemical potential $\mu=Cst-\frac{1}{\rho^2}$ when the mobility $M(\rho)$ is fixed (for instance by the previous step in a Picard's fixed-point iteration procedure).
Moreover, assuming regularity, we can derive all these quantities.
Finally, since we assume the differentials are continuous, all the implied terms are $\sL^\infty$ and integrable.
Let $h \in \mathcal{C}^4(\R;\R)$; we have the following

\begin{eqnarray*}
-\frac12 \langle M(\rho) \nabla \mu , \nabla h \rangle  
&=&-\frac12 \left\langle \frac{1}{\rho} \nabla \left(-\frac{1}{\rho^2}\right) , \nabla h \right\rangle \\
&=&-\frac12 \left\langle \frac{2}{\rho^4}\nabla \rho , \nabla h \right\rangle
\\
&=&-\frac12  \left\langle \frac{2}{3}\nabla \left(-\frac{1}{\rho^3}\right) , \nabla h \right\rangle\\
&=&-\frac13  \left\langle \frac{1}{\rho^3} , \Delta h \right\rangle.
\end{eqnarray*}

Of notice: for the last integration by parts formula (above), it is assumed that boundary term vanishes, fact which is ensured by selecting appropriate Neumann boundary conditions, which means the probability flux vanishes on the domain boundary.  Moreover
the equation \eqref{j2} is consistent with the fact that $\nabla h \cdot \vec{n}$ vanishes on the boundary.

We can also stop the calculations at the prior to the last line: in this case one observes it corresponds to considering a constant mobility $M=\frac{2}{3}$ but a chemical potential $\mu:=Cst -\frac{1}{\rho^3}$, for all constants $Cst$.
\end{proof}


Next, again based on the definition \ref{def:sol}, we give the definition of a solution based on the form suggested by previous Proposition \ref{prop:equiv} for the chemical potential.

\begin{definition}\label{def:solsimplified}
We say that $\left(\rho(t), W(t)\right)_{t\in[0,T]}$, defined on a filtered complete probability space $\left(\Omega, \mathbb{P}, \F, (\F_t)_{t\in [0,T]}\right)$, is a weak solution to \eqref{j1} on $[0,T]$ for the initial condition $\rho_{0}$, if a.s. we have:

\begin{enumerate}
\item[(b')] $\rho \in \mathcal{C}\left([0,T]\times [0,1];\mathbb{R}_+\right)\cap \mathcal{C}\left([0,T];\sV_{-1}\right)$ and $X(0) = X_{0}$,
\item[(c')] the process $\mu=Cst-1/\rho^3$ (or $\mu =Cst-1/\rho^2$) $\in \sL^1\left([0,T]; \sL^1((0,1);\mathbb{R}_+)\right)$ 
\item[(d')] $W$ is a cylindrical Wiener process in $\sL^2(0,1)$,
\item[(f')] the mobility $M$ is constant, \item[(g')] the covariance is constant (additive noise),
\item[(h')] for all $h \in D(\Delta^2)$ and for all $0 \leq  t \leq T$ :
\begin{eqnarray*}
\langle \rho(t),h\rangle &=& \langle \rho_0,h\rangle + \frac{M}{2}\int_{0}^t \left\langle \mu(r), \Delta h\right\rangle \ddr
- \int_{0}^t\left\langle \nabla h ,\dd W_r\right\rangle\\
&&- \alpha \int_{0}^t \left\langle \rho(r), \Delta^2 h\right\rangle \ddr
\end{eqnarray*}
\end{enumerate}
\end{definition}\par\medskip

Now we have the following:

\begin{theorem}\label{th:existencesimplified}
Consider the equation \eqref{j1}, with constant mobility, and constant additive noise, with initial density $\rho_0\in \mathcal{C}([0,1];\mathbb{R}^+)$ (thus nowhere vanishing) having null Neumann condition on the boundary.  Then, for all $T>0$, there exists a unique solution $(\rho_t)_{t\in[0,T]}$ in the sense of Definition \ref{def:solsimplified} up to time $T$.

However, in the case $\mu=Cst-1/\rho^2$, an additional reflection term should be added to force the solution to stay positive. It takes the form $\langle Ah, \eta\rangle$ for $\eta$ a time-space measure such that $supp(\eta)\subset\{(t,s) \in [0,T]\times (0,1) : \rho(t,s)=0 \}$.
\end{theorem}

\begin{proof}
There are two cases to address: $\mu = Cst - 1/\rho^p$, for $p=2,3$.

The proof is based on a regularization procedure to define approximated solutions of the problem.
Precisely, we have to regularize the potential $V(\rho)= Cst\ \rho+1/\rho$ (whose derivative is $\mu = Cst-1/\rho^2$) rendering it a dissipative potential $V^n$ in order to apply the classical results of existence of solution to gradient system.
The same applies to potential $V(\rho)= Cst\ \rho+1/2\rho^2$ whose derivative is $\mu = Cst-1/\rho^3$.
For these two approximated problems, the solutions will exist up to any time $T>0$.

Denoting $\Delta$ the Laplace operator on $(0,1)$ with null Neumann boundary condition and $\sV_{-1} := D((-\Delta)^{\frac12})$ the reference Hilbert space, we are in the context of gradient systems. Applying Theorem 2.2 of \cite{goudenege}, we have the existence of solution up to a reflection measure, which should force the solution to stay in the space $K$ of positive solutions.

Moreover, using Theorem 3.1 in \cite{goudenege}, in the first case $\mu:=Cst - 1/\rho^2$, since the negative power in the associated chemical potential is less than $3$, the measure does not vanish.

However, using Theorem 3.1 in \cite{goudenege}, in the second case $\mu:=Cst - 1/\rho^3$, 
since the negative power in the associated chemical potential is greater than $3$, the measure vanishes.
\end{proof}

\section{Path-wise Uniqueness}\label{sec:uniqueness}

The path-wise uniqueness will be proven in both cases with the same strategy using Gronwall's Lemma.

\begin{proposition}
Let $(\rho^{i},W)$, $i=1,2$ be any pair of weak solutions of \eqref{j1} defined on the same probability space with the same driving noise $W$, and starting from the same initial data $\rho_0$, then $\rho^1(t)=\rho^2(t)$ for all $t\geq 0$.
\end{proposition}

\Proof
In the abstract formulation, we can follow Proposition 4.2.10 of \cite{PR} to obtain 
that there exists a constant $c\in\mathbb{R}$ such that we have
\[
\E\left[\|\rho^1(t)-\rho^2(t)\|\right] \leq \exp(ct) \E\left[\|\rho^1(0)-\rho^2(0)\|\right].
\]
which is valid for a certain $c\in \R$ under the assumption of weak monotonicity.

Assuming the same initial data $\rho^{1}(0) = \rho^{2}(0)=\rho_0$, the previous result gives the uniqueness, but also a stronger result about sensitivity with respect to the initial data.

In the simplified formulation with singular potential, the uniqueness follows from dissipativity of the non-linearity.
However, in Definition \ref{def:solsimplified}, we only have continuity in the space $\sV_{-1}$, thus the norm should be taken in $\sV_{-1}$ (this space is playing the same role than the Hilbert space $\sH$ of Definition \ref{def:solabstract}).
Let $\rho(t) = \rho^1(t)-\rho^2(t)$ the difference between the two solutions starting from initial data $\rho_0$.
Due to the additive noise assumption, $\rho$ is solution of a random equation (without stochastic integral)
\begin{equation*}
\frac{\dd\rho}{\ddt}  = \frac{M}{2} \Delta (\mu^1-\mu^2) - \alpha \Delta^2 (\rho^1-\rho^2) 
\end{equation*}

Let now $\rho^N=[(-\Delta)^{-1} Proj_N] \rho$ with $Proj_N$ being a projector operator on first $N$ eigenvectors of $(-\Delta)$.
Taking the scalar product in $\sH$ with $\rho^N$ and integrating in time, we obtain
\begin{equation}\label{eq:uniqueness}
|\rho^N(t)|^2_{-1} =
|\rho^N(0)|^2_{-1}
- M \int_0^t \langle \mu^1(r)-\mu^2(r), \rho^N(r) \rangle \ddr  
-2\alpha \int_0^t \langle (-\Delta)^{1/2} \rho^N(r), (-\Delta)^{1/2} \rho^N(r) \rangle \ddr.
\end{equation}
Since $\mu^i$, $i=1,2$ are in $\sL^1$, we have
for almost all $r\in[0,t]$, 
\begin{equation}\label{eq:rhorhoN}
\langle \mu^1(r)-\mu^2(r), \rho(r) - \rho^N(s)\rangle
\leq \|\rho(r) - \rho^N(r)\|_{\sL^\infty([0,1])}\|\mu^1(r)-\mu^2(r)\|_{\sL^1([0,1])},
\end{equation}
where $\|\cdot\|_{\sL^\infty([0,1])}$ and $\|\cdot\|_{\sL^1([0,1])}$ are the classical Lebesgue norms on the space $[0,1]$. The latter term converges to zero since $\rho^N(r)$ converges uniformly to $\rho(r)$ on $[0,1]$.
Moreover, for both cases $\mu=Cst-1/\rho^p$ with $p=2,3$, we have 
\begin{eqnarray}\label{eq:negativity}
&-&\langle \mu^1(r)-\mu^2(r), \rho(r) \rangle
= \left\langle \frac{1}{(\rho^1(r))^p}-\frac{1}{(\rho^2(r))^p} , \rho(r)\right\rangle\nonumber\\
&\leq & \left\langle \frac{(\rho^2(r)-\rho^1(r))((\rho^2(r))^{p-1}+(p-2)\rho^1(r)\rho^2(r)+(\rho^1(r))^{p-1})}{(\rho^1(r)\rho^2(r))^p}, \rho^1(r)-\rho^2(r)\right\rangle\nonumber\\
&\leq & 0
\end{eqnarray}

Injecting \eqref{eq:rhorhoN} into \eqref{eq:uniqueness}, using the sign coming from \eqref{eq:negativity}, 
taking the limit in \eqref{eq:uniqueness} as $N$ grows to infinity, we obtain for all $t\geq0$:
\[
|\rho^1(t)-\rho^2(t)|_{-1}^2 \leq |\rho^1(0)-\rho^2(0)|_{-1}^2 =0,
\]
thus $\rho^1(t)=\rho^2(t)$ for all $t\geq0$.
\qquad
\EProof

\section{Convergence of Invariant Measures}\label{sec:invariant}

We say that the measure $\Pi$ is an invariant measure for the stochastic PDE (or the associated semi-group) if, for a random variable $\rho_{0}$, whose distribution is $\Pi$, the solution $(\rho(t))_{0\leq t \leq T}$ is such that its distribution at time $t$ is again $\Pi$.
Such an invariant measure can be easily calculated (for instance in finite dimensional case) as eigenvector of an adjoint operator $\mathcal{P}^{*}$ of the semi-group $\mathcal{P}$ of the Markovian evolution.

Usually, we prove that an approximated version (described by some index $n\in\N$) of the equation possesses a semi-group which is Strong Feller (see .e.g \cite{PeszatZabczyk} for details) and irreducible. By Doob theorem (see \cite{doob}) we deduce that there exists a unique and ergodic invariant measure $\Pi^n$ defined on the abstract Hilbert space $\sH$.

Denoting $\rho^n(t;x)$ the solution starting form the initial data $x\in\sH$, we can define for all Borel bounded test functions $\phi\in\mathcal{B}_{b}(\sH)$ the quantity
\[
\mathcal{P}^{n}_{t}\phi(x):=\mathbb{E}[\phi(\rho^n(t;x))]
\]
which is the semi-group of the Markovian evolution.

For instance, using Bismut-Elworthy-Li formula (assuming inversibility of covariance operator and other technical assumptions), it can be seen that $(\mathcal{P}^{n}_t)_{t\ge 0}$ satisfies, for all $\phi \in \mathcal{B}_{b}(\sH)$,  $n\in\N$ and $t >0$, the following estimate

\beq\label{Eq:1.3}
| \mathcal{P}^{n}_{t}\phi(x) - \mathcal{P}^{n}_{t}\phi(y)| \leq C(n) \frac{e^{\lambda t}}{\sqrt{t}}\|\phi\|_\infty \|x-y\|_{\sH}, \quad\text{for all } x,y \in \sH,
\eeq
where  we expect that the constant $C(n)$ could be independent of the parameter $n\in\N$.

Furthermore, the irreducibility is usually proven using a control-type argument, which is quite clear in the case of (non-degenerate) additive noise.

If the variable $\lambda$ is non-positive, then there is exponential contraction, and the law of the solution converges exponentially fast to the invariant measure.

\subsection{Regular abstract formulation}
\label{subsec:regular}

Thanks to the previous Section, we have the existence of a unique solution in the regular abstract formulation.
For the second part of the Theorem \ref{th:main1} one needs to prove there exists an invariant measure to the equation, and that the solution converges to this invariant measure as $T\rightarrow +\infty$.
For this, we will need supplementary assumptions regarding the abstract formulation, precisely a ``strict monotonicity'', described by the following Lemma:

\begin{lemma}[Lemma 4.3.8 from \cite{PR} - strict monotonicity]
Assume that the constant $c=-\lambda$ in the weak monotonicity assumption for some $\lambda \in (0,+\infty)$.
Let $\eta\in(0,\lambda)$.  Then there exists $\delta_\eta\in(0,+\infty)$ such that for all $v\in V$
\[
2\,_{\sV^*}\langle A(v),v\rangle_{\sV} + \|B(v)\|^2_{HS} \leq -\eta \|v\|^2_{\sH} + \delta_\eta.
\]
\end{lemma}

The above is enough to obtain the following:

\begin{theorem}
Assume strict monotonicity.  Then there exists an invariant measure $\Peq$ for the semi-group $\mathcal{P}_t$ such that
\begin{equation}\label{eq:integrabilitymeasure}
\int_\sH \|x\|^2\Peq(\ddx) < + \infty.
\end{equation}
Moreover, for any $\phi:\sH\rightarrow \R$ Lipschitz, for any initial data $\rho\in \sH$ and any invariant measure $\Peq$ for the semi-group $\mathcal{P}_t$, for all $t\in [0,+\infty)$
\[
\left| \mathcal{P}_t \phi(\rho) - \int_\sH \phi(x) \Peq(\dd x) \right| \leq \exp^{-\frac{\lambda}{2}t} \|\phi\|_{Lip} \int_\sH \|\rho-x\|_H\Peq(\dd x)
\]
In particular, there exists exactly one invariant measure for the semi-group $\mathcal{P}_t$ with the property that $\int_\sH\|x\|\Peq(\ddx)<+\infty$.
\end{theorem}
\begin{proof}
The proof of this theorem comes from Theorem 4.3.9 of \cite{PR}.  
\end{proof}

\subsection{Simplified version with singular potential}
\label{subsec:simplified}

To prove the existence of invariant measure, first, we need to have a reference measure in the abstract Hilbert space.
It is usually given by the asymptotic law of the solution of the linear part of the equation.
For instance, if the operators $A$ and $B$ are linear and constant, we can define explicitly the solution $Z(t;x)$ of the linear equation with initial data $x \in \sH$ which is given by
\beq
Z(t;x) = e^{-tA/2}x + \int_0^te^{-(t-s)A/2}B\dd W_s.
\eeq
As easily seen this process is in $\mathcal{C}([0,+\infty[;\sH)$ (see \cite{MR1207136}). In particular, the mean of $Z$ is constant and the law of the process $Z(t;x)$ is the Gaussian measure on $\sH$:
\[
\gamma(t)=\mathcal{N}\big(e^{-tA/2}x,\Sigma_t^2\big),\text{ with } 
\Sigma_t^2 = \int_0^te^{-sA/2}BB^*e^{-sA/2}\dds 
\]
with $\mathcal{N}(m,\Sigma^2)$ the standard notation for the Normal Law of mean $m$ and covariance operator $\Sigma^2$.
If we let $t\rightarrow +\infty$, the law of $Z(t;x)$ converges to the Gaussian measure on $\sH$:
\[
\gamma := \mathcal{N}\left(0,\int_0^\infty e^{-sA/2}BB^*e^{-sA/2}\dds \right).
\]
When $B$ commutes with $A$, we can obtain an explicit representation of the covariance operator $\Sigma^2 := B(-A)^{-1}B^*$.

What about the invariant measure of the non-linear equation?
It is classic knowledge that if the equation \eqref{j1} may be approximated by a gradient system in $\sH$ which can be rewritten as:
\beq
\left\{\begin{array}{l}
\dd X_t^n+\left(\frac{1}{2}A^n X_t^n + \nabla V^n(X_t^n)\right)\ddt = B\dd W_t,\\
\\
X_0^n= x \in \sH,
\end{array}\right.
\eeq

where $\nabla$ denotes the gradient in the Hilbert space $\sH$,
then one deduces there exists an unique and ergodic Gibbs-type invariant measure $\Pi^n$ given by:
\[
\Pi^n(\ddx)=\frac{1}{Z^n}\exp(-\mathcal{E}^n(x))\gamma(\ddx),
\]
where $Z^n$ is a normalization constant and $\mathcal{E}^n$ is  the energy related to the potential $V^n$ (as in \eqref{j3}).
We need finish the proof by showing that the sequence $(\Pi^n)_{n\in\N}$ converges to the measure
\[
\Peq(\ddx)=\frac{1}{Z}\exp(-\mathcal{E}(x))\1_{x\in K}\gamma(\ddx),
\]
where for all $x \in L^2(0,1)$, we have
$\mathcal{E}(x):= \int_{0}^{1}V(x(s))\dds$ (as in \eqref{j3} but on fixed domain $(0,1)$).

Of course, we need a strong assumption on the limit measure $\Peq$, precisely we ask that the quantity $\mathcal{E}(x)$ is well-defined for $\gamma$ almost all $x\in \sH\cap K$ (but possibly infinite with the convention that $\exp(-\infty)=0$).

In the previous section, one can notice that the integrability condition \eqref{eq:integrabilitymeasure} of the invariant measure $\Peq$ is deduced from the strict monotonicity.

Naturally, we expect that the energy is well-posed on the space of non-negative functions
\[
K=\{x\in \sL^2 :\;  x \ge 0\}
\]
assuming we have a non-negative potential $V$.

\begin{proposition}\label{prop:conv_measure}
Assume that $V^n$ is non-negative, non-increasing and converges from below to $V$ on $\R_+$, but converges to $+\infty$ on $\R_-$.  Then  we have the convergence
$
\nu^n \mathop{\rightharpoonup}_{n\rightarrow+\infty} \nu.
$
\end{proposition}

\Proof
Let $\psi \in \mathcal{C}_{b}(\sL^2(0,1),\R)$. We want to prove that 
\beq\label{eq:convergencemeasure}
\int_{\sH} \psi(x)\exp({-\mathcal{E}^n(x)})\gamma(\ddx)\quad  \mathop{\longrightarrow}_{n\rightarrow+\infty} \quad \int_{\sH} \psi(x)\exp({-\mathcal{E}(x)})\1_{\rho\in K}\gamma(\ddx).
\eeq
We first prove that for ($\gamma$-almost surely) $x\in\sH$ we have
\beq\label{eq:convergenceenergy}
\exp({-\mathcal{E}^n(x)}) \quad \mathop{\longrightarrow}_{n\rightarrow+\infty} \quad \exp({-\mathcal{E}(x)})\1_{x\in K}.
\eeq
Since the reference measure $\gamma(C([0,1]))=1$, we can restrict our attention to $x\in C([0,1])$. 
Then if $x\notin K$ there exists $\delta_{x}>0$ such that $Leb(\{s\in [0,1] : x(s)\leq -\delta_{x}\})>0$, $Leb$ being the Lebesgue measure.
Then, since $V^n(x)$ is non-negative and non-increasing on $\R_+:=(0,\infty)$, we have:
\begin{eqnarray*}
0\leq\exp({-\mathcal{E}^n(x)}) &\leq& \exp\left(-\int_{0}^1 V^n(x(s))
1_{\{x\leq -\delta_{x}\}}\dds\right)\\
&\leq& \exp\left(-\int_{0}^1 V^n(-\delta_{x})1_{\{x\leq -\delta_{x}\}}\dds\right)\\
&\leq& \exp\left(-V^n(-\delta_{x})Leb(\{x\leq -\delta_{x}\})\right).
\end{eqnarray*}
By assumption, this latter term converges to zero as $n$ grows to infinity.\\
Now, for $x \in K$, $V^n(x(s))$ converges to $V(x(s))$ almost everywhere as $n$  grows to infinity. Moreover $0\leq V^n(x(s)) \leq V(\sup\|x\|)$, and by the Dominated Convergence Theorem we deduce \eqref{eq:convergenceenergy}.
Finally, \eqref{eq:convergencemeasure} follows again by the Dominated Convergence Theorem.\qquad \EProof

We know (see \cite{MR2349572, goudenege}) that the approach described in the previous section will also work in the singular context.
Moreover the semi-group $\mathcal{P}_{t}:=\lim_{n\rightarrow+\infty} \mathcal{P}^n_{t}$ is Strong Feller, such that
for all $\phi \in \mathcal{B}_{b}(\sH)$, for all $x,y \in \sH$ and $t >0$: 
\beq\label{eq:strongfeller}
| \mathcal{P}_{t}\phi(x) - \mathcal{P}_{t}\phi(y)| \leq \frac{\|\phi\|_\infty}{\sqrt{t}} \|x-y\|_{\sH}.
\eeq





\section{Concluding With A Conjecture}\label{concl}

We took on to proving the existence and uniqueness of solutions to a previously unstudied class of non-linear stochastic partial differential equations with a singular coefficient that has recently been introduced ~\cite{jay3} within the framework of polymer dynamics theories geared towards polymer physics and rheology.

The equations under scrutiny in this paper are simplified versions of the originally suggested for being studied.  
We have shown that there exists solutions to these approximated problems under classical assumptions and abstract formulation.
It appears at this stage that it is plausible to obtain existence of solution results for the more general singular problem just alluded to by extracting properly defined convergent sequences of approximated solutions, and this upon obtaining enough a priori estimates, and passing afterwards  to the limit.
For instance, using an $\varepsilon$ approximation (by Yosida regularization, or globally Lipschitz potential, or dissipative potential), one may achieve the results by compactness related arguments.
Using Picard's fixed-point iteration procedure, it seems also plausible to obtain convergence to a solution.

We can state the following conjecture based on the existence of an abstract formulation which encompasses the space of functions and the non-linearity. 
Let $\mathcal{H}$ be this hypothetical, abstract space of functions and assume there exists a potential derivative $V'$, a mobility $M$ and a covariance operator $Q$ well-defined on $\mathcal{H}$.

Then the following conjecture  remains to be proved:

\begin{theorem}[{\bf Conjecture}]\label{th:conjecture}
Let $\rho_{0}$ be a smooth enough initial data.
There exists a process $\left(\rho(t), W(t)\right)_{t\in[0,T]}$, defined on a filtered complete probability space $\left(\Omega, \mathbb{P}, \F, (\F_t)_{t\in [0,T]}\right)$ which is a solution to \eqref{j1} on $[0,T]$ in the sense that
\begin{enumerate}
\item[(b'')] $\rho \in \sL^2\left(\Omega; \mathcal{H}\right)$ and $\rho(0) = \rho_{0}$,
\item[(c'')] $\mu := V'(\rho) \in \sL^2\left(\Omega; \mathcal{H}\right)$,
\item[(d'')] $W$ is a continuous version of Brownian motion,
\item[(e'')] the process $\left(\rho(t),W(t)\right)$ is $(\F_t)$-adapted,
\item[(f'')] $M(\rho) \in \sL^2\left(\Omega; \mathcal{H}\right)$,
\item[(g'')] $Q(\rho) \in \sL^2\left(\Omega; \mathcal{H}\right)$,
\item[(h'')] for all $0 \leq  t \leq T$:
\begin{eqnarray*}
\rho(t) &=& \rho_0 + \frac12 \int_{0}^t  \nabla (M(\rho(r)) \nabla \mu(r)) \ddr + \int_{0}^t Q(\rho(r))\dd W_r.
\end{eqnarray*}
\end{enumerate}
Moreover, for $T\rightarrow +\infty$, the distribution of the solution converges to an equilibrium density.
\end{theorem}

\end{document}